\documentclass[reqno]{amsart}
\usepackage[foot]{amsaddr}
\pdfoutput=1

\usepackage{amsmath}
\usepackage{amssymb}
\usepackage{booktabs}
\usepackage{rotating}
\usepackage{color}
\usepackage{enumitem}
\usepackage[T1]{fontenc}
\usepackage[utf8]{inputenc}
\usepackage[l2tabu,orthodox]{nag}
\usepackage{tensor}
\usepackage[vcentermath]{youngtab}
\usepackage{nicefrac}
\usepackage{mathrsfs} 

\usepackage{dsfont}
\usepackage{mathbbol}

\usepackage{todonotes}

\usepackage[
		pdftex,
		bookmarks,
		bookmarksopen,
		bookmarksopenlevel=1,
		bookmarksnumbered,
		breaklinks,
		colorlinks,
		linkcolor=linkcolor,
		citecolor=citecolor,
		urlcolor=urlcolor,
	]{hyperref}

\hypersetup{
	pdftitle    = {Abundant superintegrable systems and Hessian structures},
	pdfauthor   = {John Armstrong and Andreas Vollmer},
	pdfsubject  = {article},
	pdfcreator  = {\LaTeX{} with `amsart' class},
	pdfproducer = {pdf\LaTeX},
	pdfkeywords = {superintegrable systems}
}

\definecolor{linkcolor}{rgb}{0.5,0.0,0.0}
\definecolor{citecolor}{rgb}{0.0,0.5,0.0}
\definecolor{urlcolor} {rgb}{0.0,0.0,0.5}

\title[Abundant superintegrability and Hessian structures]
	{Abundant superintegrable systems\\ and Hessian structures}

\author{John Armstrong}
\email{john.armstrong@kcl.ac.uk}
\address{Department of Mathematics, King's College London, London WC2R 2LS, UK}

\author{Andreas Vollmer}
\email{andreas.vollmer@uni-hamburg.de}
\address{Fachbereich Mathematik, University of Hamburg,	20146 Hamburg, Germany}

\numberwithin{equation}{section}

\newtheorem{theorem}{Theorem}[section]
\newtheorem{proposition}[theorem]{Proposition}

\theoremstyle{definition}
\newtheorem{definition}[theorem]{Definition}
\theoremstyle{remark}
\newtheorem{remark}[theorem]{Remark}
\newtheorem{example}[theorem]{Example}
\newtheorem{observation}[theorem]{Observation}

\setlist[enumerate,1]{label=(\roman*)}

\newcommand{\RR}{\ensuremath{\mathds{R}}}

\newcommand{\grad}{\ensuremath{\mathrm{grad}}}
\newcommand{\Sym}{\ensuremath{\mathrm{Sym}}}
\newcommand{\trace}{\ensuremath{\mathrm{trace}}}
\newcommand{\Span}{\ensuremath{\mathrm{span}}}

\begin{document}

\begin{abstract}
    We show that a large class of non-degenerate second-order (maxi\-mally) superintegrable systems gives rise to Hessian structures, which admit natural (Hessian) coordinates adapted to the superintegrable system.
    In particular, abundant superintegrable systems on Riemannian manifolds of constant sectional curvature fall into this class. We explicitly compute the natural Hessian coordinates for examples of non-degenerate second-order superintegrable systems in dimensions two and three.
\end{abstract}

\maketitle
\tableofcontents

%============================================================================%
\section{Introduction}

Let $(M,g)$ be a (smooth) Riemannian manifold (the setup in this paper is of a local nature). A \emph{Hessian structure} on $(M,g)$ is given by a flat torsion-free connection~$D$ such that
\begin{equation*}
	g = Dd\psi
\end{equation*}
for a function $\psi\in\mathcal C^\infty(M)$ on $M$. The function $\psi$ is called \emph{Hessian potential} of $(M,g,D)$. The Hessian potential is not unique, as after adding any solution $f$ of $Ddf=0$, another Hessian potential for the same Hessian structure is obtained.
It was recently shown by R.~Bryant that metrics in dimension~$2$ are always locally hessianisable, i.e.~of the above form \cite{bryant2024}, see also \cite{Amari&Armstrong} where this result was independently obtained for analytic metrics.

\begin{remark}
	In \cite{Noguchi}, \emph{statistical manifolds} $(M,g,D^\pm)$ are introduced as defined by a Riemannian manifold $(M,g)$ endowed with a pair of torsion-free connections $D^\pm$ that are dual to each other, i.e.\ they satisfy
	\begin{equation*}
		X(g(Y,Z))=g(D^+_XY,Z)+g(Y,D^-_XZ)\,,
	\end{equation*}
	where $X,Y,Z\in\mathfrak X(M)$.
	Similarly, a Riemannian manifold $(M,g)$ endowed with a pair $D^\pm$ of dual connections that are both flat is called a \emph{$g$-dually flat structure} \cite{Amari&Armstrong}.
	We may thus interpret a Hessian structure $(M,g,D)$ as a $g$-dually flat, statistical manifold. Indeed, the dual connection $D^*$ of $D$ is also flat (and torsion-free). Given a $g$-dually flat, statistical manifold, we obtain a pair of associated Hessian structures.
\end{remark}

The (locallly defined) Hessian potentials of the connections $D$ and $D^*$ of a Hessian structure are (locally) linked by a Legendre-Fenchel transform, which is well-defined on a convex set $\mathbb X$, since $g$ is smooth and positive definite.
Denote by~$\mathbb X^*$ the dual space of $\mathbb X$ and by $\langle\cdot,\cdot\rangle:\mathbb X^*\times\mathbb  X\to\RR$ the natural pairing. The Legendre-Fenchel transform of $f:\mathbb X\to\RR$ then is the function $F:\mathbb Y\to\RR$,
\[
	F(y):=\sup_{x\in\mathbb  X}(\langle y,x\rangle-f(x))
\]
defined on the set
$ \mathbb Y=\left\{x\in \mathbb X^*\,:\,\sup_{x\in\mathbb X}(\langle y,x\rangle-f(x))<\infty\right\}\subset \mathbb X^* $.

A Hessian structure $(M,g,D)$ induces a distinguished family of local coordinate systems: for a point $p\in M$, consider normal coordinates around $p$, i.e.\ given by solutions of the geodesic equation via the exponential mapping around $p$. Such coordinates are called \emph{Hessian coordinates}. Obviously, Hessian coordinates are not unique.\medskip

Let $(M,g)$ be a simply connected (smooth) Riemannian manifold, and consider the function $H:T^*M\to\RR$,
$$ H(x,p) = g^{ij}(x)p_ip_j+V(x) $$
(Einstein's convention applies), called the \emph{Hamiltonian}, where $(x,p)$ denote canoni\-cal Darboux coordinates on $T^*M$. The function $V=V(x)$ is called the \emph{potential}.
The Hamiltonian $H$ is said to be (maximally) \emph{superintegrable} if there exist $2n-2$ functions $F^{(\alpha)}:T^*M\to\RR$, $1\leq\alpha\leq 2n-2$, such that
\begin{equation}\label{eq:integral}
	\{H,F^{(\alpha)}\}=0
\end{equation}
and such that $(F^{(\alpha)})_{0\leq\alpha\leq 2n-2}$ are functionally independent, where we agree that $F^{(0)}=H$.
A \emph{superintegrable system} defined by $(H=F^{(0)},\dots,F^{(2n-2)})$ is called \emph{second-order} if the integrals $F^{(\alpha)}$ can be chosen in the form
\[ F(x,p) = K^{ij}(x)p_ip_j+W(x)\,, \]
where we omit the superscript $(\alpha)$ for better legibility.
%
%From now on we will exclusively speak about second-order (maximally) superintegrable systems, and therefore suppress these attributes for brevity, whenever there is no risk of confusion.
%
Note that for second-order systems, \eqref{eq:integral} is an inhomogeneous cubic polynomial in momenta whose homogeneous parts have to be satisfied independently. Its leading part is equivalent to the condition that $K_{ij}=g_{ia}g_{jb}K^{ab}$ are components of a Killing tensor field. Its linear homogeneous component, on the other hand, is equivalent to the condition
\begin{equation}\label{eq:pre-BD}
	dW = \hat K(dV)\,,
\end{equation}
where we denote by~$\hat K$ the endomorphism associated to the Killing tensor $K$ via~$g$.
We will tacitly use the circumflex with this meaning from now on and omit the diacritic when using index notation.
Applying the differential to~\eqref{eq:pre-BD}, we have the equation
\begin{equation}\label{eq:BD}
	d(\hat K(dV))=0 \,,
\end{equation}
called \emph{Bertrand-Darboux} condition \cite{bertrand_1857,darboux_1901}. Note that ~\eqref{eq:BD} holds for any endomorphism $\hat K$ associated with an integral $F^{(\alpha)}$, $0\leq\alpha\leq 2n-2$, and moreover for any element in the linear span $\mathcal K$ of these endomorphisms. By a slight abuse of notation, we denote the space of associated Killing tensors by the same symbol $\mathcal K$.
Following~\cite{KSV2023,KSV2024}, we now introduce \emph{irreducible} superintegrable systems as second-order (maximally) superintegrable systems such that $\mathcal K$ forms an irreducible set, i.e.\ a set of endomorphisms that do not have common eigenvectors.
\begin{definition}
	Let $(M,g)$ be a Riemannian manifold. The tuple $(M,g,\mathcal V,\mathcal K)$ is then called an \emph{irreducible second-order (maximally) superintegrable systems} (or \emph{irreducible system} for short) if
	\begin{enumerate}[label=(\roman*)]
		\item $\mathcal K$ forms an irreducible linear subspace in the space of Killing tensors of $g$ of dimension $\dim(\mathcal K)\geq 2n-1$,
		\item for any $V\in\mathcal V$ and $\hat K\in\mathcal K$, \eqref{eq:BD} holds,
		\item there exist $K^{(\alpha)}\in\mathcal K$, $0\leq\alpha\leq 2n-2$, with $K^{(0)}=g$ and such that $(K^{(\alpha)})_{0\leq\alpha\leq 2n-2}$ are functionally independent functions $T^*M\to\RR$.
	\end{enumerate}
\end{definition}
Note that this definition is consistent with the previous definition of a second-order (maximally) superintegrable systems, since \eqref{eq:pre-BD} allows us to reconstruct the integrals $F^{(\alpha)}$ in that definition.
It is proved in \cite{KSV2023} that for \emph{irreducible} superintegrable systems a potential $V\in\mathcal V$ satisfies
\begin{equation}\label{eq:wilczynski}
	\nabla^2V=\hat T(dV)+\frac1n\,g\,\Delta V\,,
\end{equation}
where $\hat T\in\Gamma(\Sym_2(T^*M)\otimes TM)$ is the structure tensor, and where $\nabla^2$ and $\Delta$ denote the Hessian and, respectively, the Laplace-Beltrami operator with respect to the Levi-Civita connection of $g$. It follows that $\hat T(dV)$ is trace-free (with respect to $g$). It determines an associated 1-form $t\in\Omega^1(M)$, defined by
$$ t = \trace(\,\hat T\,)\,. $$
For later reference, we also introduce the tensor field $T\in\Gamma(\Sym_3(T^*M))$ associated to $\hat T$ via the metric $g$,
\[
	T(X,Y,Z)=\hat T(X,Y)(g(Z))
\]
for $X,Y,Z\in\mathfrak{X}(M)$. By a slight abuse of terminology, we also call $T$ structure tensor, along with $\hat T$, as there is no risk of confusion.
From now on we will exclusively speak about irreducible second-order (maximally) superintegrable systems, and therefore suppress these attributes for brevity whenever there is no risk of confusion.

%============================================================================%
\section{Non-degenerate systems and Hessian structures}

An irreducible superintegrable system is called a \emph{non-degenerate} system if the space $\mathcal V$ of potentials has dimension $n+2$, $\dim(\mathcal V)=n+2$. Note that in the case of an analytic metric (and hence analytic Killing tensors), this dimension coincides with the number of integration constants in \eqref{eq:wilczynski}. This allows one to choose, at a point $p\in M$, the values of $V(p)$, $\nabla V(p)$ and $\Delta V(p)$ freely.
In practical computations, $\mathcal V$ is typically presented as a $(n+2)$-parameter linear family of potentials.
We discuss some results of~\cite{KSV2023,KSV2024}, and we then prove that non-degenerate systems often induce Hessian structures. In the next section, we will then focus on a particular case, namely abundant non-degenerate systems on Riemannian manifolds of constant sectional curvature in dimension $n\geq3$.
On such spaces, the existence of a Hessian structure can be deduced immediately from a statement proved in \cite{KSV2023}.
%Consider the affine connection
%$$ \nabla^{\pm A} := \nabla^g \pm A, $$
%where $A\in\Gamma(\Sym_2(T^*M)\otimes TM)$,
%\begin{equation}\label{eq:A.from.T}
%	A := \frac13 \left( T + \frac{1}{n-1}\, g\otimes t \right).
%\end{equation}

We discuss some definitions and results in~\cite{KSV2023,KSV2024} that allow us to draw conclusions regarding Hessian structures.
Let $(M,g)$ be a Riemannian manifold of dimension $n\geq3$ endowed with a non-degenerate system with structure tensor $T=T_{ijk}dx^i\otimes dx^j\otimes dx^k$.
It is shown in the reference that its components satisfy
\begin{equation}\label{eq:linear}
	T_{ijk}=S_{ijk}+\frac{n}{(n-1)(n+2)}\left( t_ig_{jk}+t_jg_{ik}-\frac2n\,g_{ij}t_k\right)\,,
\end{equation}
where $S_{ijk}$ are the components of a totally symmetric and trace-free tensor field~$S$. For later reference we also introduce the $(2,1)$-tensor field $\hat S\in\Gamma((T^*M)^{\otimes2}\otimes TM)$ associated to $S$ via the inverse metric by
\[
	\hat S(X,Y)(\alpha) = S(X,Y,g^{-1}\alpha)
\]
for $X,Y\in\mathfrak X(M)$ and $\alpha\in\Omega^1(M)$.
Following \cite{KSV2023}, we also introduce the totally symmetric tensor field $A\in\Gamma(\mathrm{Sym}_3(T^*M))$ that satisfies
\begin{equation}\label{eq:A.from.T}
	A = \frac13\left( T+\frac{1}{n-1}\,g\otimes t\right).
\end{equation}
%$$ A_{ijk} := \frac13 \left( T_{ijk} + \frac{1}{n-1} g_{ij} t_k \right). $$
We write $\hat A:=g^{-1}A$ for the corresponding $(1,2)$-tensor field.
Subsequently, we can introduce a pair of torsion-free connections by
\begin{equation}\label{eq:nabla+-}
	\nabla^{\pm\hat A} := \nabla^g \pm \hat A\,.
\end{equation}
We will also write $\nabla^{\pm}$ instead of $\nabla^{\pm\hat A}$, respectively, if there is no risk of confusion.
We now prove the following theorem for non-degenerate systems.
For a more concise notation, we introduce the 1-form $\bar t\in\Omega^1(M)$ associated to $t$ via
\begin{equation}
	\bar t=\frac{n}{(n-1)(n+2)}\,t\,.
\end{equation}
\begin{theorem}\label{thm:main}
	Let $(M,g)$ be a conformally flat Riemannian manifold of dimension~$n\geq3$, which we assume to be endowed with a non-degenerate system with structure tensor $T$. Furthermore, we assume that
	\begin{subequations}\label{eqs:hessian.conditions}
	\begin{align}
		%\mathring Z_{ij} &= 0,
		\mathring{R}_{ij} &= \frac19\,
			\Pi_\circ\left(
			\mathring T\indices{_{i}^{ab}}\mathring T_{jab}-(n-2)(\mathring T_{ija}\bar t^a+\bar t_i\bar t_j)
		\right)
		\label{eq:hessian.condition.1}
		\\
		R &= \frac19\left( |\mathring{T}|^2-(n+2)(n-1)|\grad_g\bar t|^2 \right)
		\label{eq:hessian.condition.2}
	\end{align}
	\end{subequations}
	where $\mathring{R}_{ij}$ denotes the trace-free part of the Ricci tensor and $R$ the scalar curvature, and where $\Pi_\circ$ is the projection onto the trace-free part (with respect to $g$).
%	with $\mathring Z_{ij}=Z_{ij}-\frac1ng_{ij} Z\indices{^a_a}$, where
%	\[
%		Z_{ij} = \mathring T\indices{_{i}^{ab}}\mathring T_{jab}-(n-2)(\mathring T_{ija}\bar t^a+\bar t_i\bar t_j)-R_{ij}\,.
%	\]
	Then the underlying metric $g$ is Hessian.
\end{theorem}
\noindent We remark that the statement is trivial in dimension~$2$ \cite{bryant2024,Amari&Armstrong}, and hence the requirement $n\geq3$ is of a technical nature.
\begin{proof}
	We use Proposition~4.4 of~\cite{KSV2023}, where it is shown that the integrability conditions for~\eqref{eq:wilczynski} under the hypothesis are equivalent to the conditions~\eqref{eq:linear},
	\begin{align}
		\Pi_\circ\left( T\indices{^a_{ik}}T_{ajl}-T\indices{^a_{jk}}T_{ail}\right) &= 0
		\\
		\Pi_\mathrm{Weyl} T\indices{^a_{ik}}T_{ajl} &= 0 %{\young(ik,jl)}_\circ^*
		\label{eq:integrability.quadratic}
		\\
		\Pi_{ \resizebox{7pt}{!}{\yng(3,1)} }\ T_{ijk,l} &= 0\,. %\young(kji,l)
		\label{eq:integrability.differential}
	\end{align}
	Here, $\Pi_\mathrm{Weyl}$ is the projector onto the Weyl symmetric part, and $\Pi_{ \resizebox{7pt}{!}{\yng(3,1)} }$ is the projector that first skew symmetrises in $k,l$ and then symmetrises in $i,j,k$.
%	In order to ensure a concise notation, we here use abstract index notation for the tensor fields $\hat A=A_{ij}^{k}dx^i\otimes dx^j\otimes \partial_k$ and $\hat T=T\indices{_{ij}^k}dx^i\otimes dx^j\otimes \partial_k$, and for the 1-form $t=t_kdx^k$. Note that when using abstract index notation, we shall tacitly raise and lower indices using the metric $g$ in the obvious manner.
	Rewriting the equations~\eqref{eq:integrability.quadratic} and~\eqref{eq:integrability.differential} in terms of $A_{ijk}$, we obtain, c.f.~\cite{KSV2023},
	\begin{align}
		\Pi_\circ\left( A^a_{ik}A_{ajl}-A^a_{jk}A_{ail} \right) &= 0
		\\
%		\young(kji,l)\,A_{ijk,l} &= 0\,
		\nabla_lA_{ijk}-\nabla_kA_{ijl} &= 0\,,
	\end{align}
	(compare \cite{KSV2023}, where a similar computation is performed in the specific case of constant curvature).
	Recalling~\eqref{eq:linear} and~\eqref{eq:A.from.T}, a direct computation shows
	\[
		9 \left( A^{a}_{i b} A^{b}_{j a}-A_{a}^{a b} A_{b i j} \right)
		= S\indices{_i^{ab}}S_{jab}-(n-2)\left(S_{ija}\bar t^a+\bar t_i\bar t_j\right) + ng_{ij}\,|\bar t|^2\,.
	\]
	Due to~\eqref{eq:hessian.condition.1}, this implies, for the Ricci tensor of $g$,
	\[
		R_{ij} = A^{a}_{i b} A^{b}_{j a}-A_{a}^{a b} A_{b i j}\,.
	\]
	We conclude that the two connections $\nabla^{\pm}=\nabla^{\pm\hat A}$ are both flat, and that $(M,g,\nabla^\pm)$ hence  are Hessian structures.
\end{proof}

%============================================================================%
\section{Hessian potentials and structure functions of abundant systems}

The hypotheses of Theorem~\ref{thm:main} apply, in particular, to abundant systems on Riemannian manifolds of constant sectional curvature in dimension $n\geq3$.
A non-degenerate system is called an \emph{abundant} system if \eqref{eq:BD} holds for any $V\in\mathcal V$ and for any $K\in\mathcal K^\text{abt}$ where $\mathcal K^\text{abt}$ is a linear space of Killing tensors of dimension $\dim(\mathcal K^\text{abt})=\frac{n(n+1)}{2}$. Non-degenerate systems in dimension $n\in\{2,3\}$ are necessarily abundant \cite{KKM05c}, and in fact all known examples of non-degenerate systems (for any $n$) are abundant ones.

In the case of abundant systems in dimension $n\geq3$, the statement of Theorem~\ref{thm:main} follows immediately from \cite{KSV2023}.
This case is also the main application of Theorem~\ref{thm:main}. Indeed, the statement of Theorem~\ref{thm:main} is obtained from Corollary~7.17 in~\cite{KSV2023}, which states, for an abundant system on a Riemannian manifold of constant sectional curvature in $n\geq3$, that the torsion-free and Ricci-symmetric connection $\nabla^{-\hat A}$ is flat, where $\hat A$ is defined as above in~\eqref{eq:nabla+-}. We observe that it follows that $\nabla^{+\hat A}$ is flat as well. Therefore $(M,g,\nabla^{\pm\hat A})$ defines a pair of Hessian geometries on $M$.

\begin{example}\label{ex:canonical.hessian.structure}
	Let $(M^n,g)$ be a Riemannian manifold with dimension $n\geq2$ and of constant curvature $\kappa$ equipped with an abundant superintegrable system with structure tensor $T$ (we define $\hat A$ as above).
	Then $(M,g)$ is endowed with the flat $g$-dual structure $\nabla^\pm=\nabla^{\pm\hat A}$, and therefore carries a Hessian structure.
	Note that in this example we have included the dimension $n=2$, since any surface carries, at least locally, a Hessian structure \cite{bryant2024,Amari&Armstrong}.
\end{example}

From now on, we assume the Riemannian manifold $(M,g)$, $n\geq2$, to be of constant sectional curvature $\kappa$, and endowed with an abundant system.
It is endowed with the pair of flat connections $\nabla^\pm$, giving rise to a pair of Hessian potentials $\psi_\pm$ that satisfy
\begin{subequations}\label{eq:hessian}
\begin{align}
	g &= \nabla^{+}d\psi_+ = \nabla^{-}d\psi_-\,,
	\label{eq:g.hessian}
	\\
	\label{eq:A.hessian}
	A &= \frac12(\nabla^{+})^3\psi_+ = -\frac12(\nabla^{-})^3\psi_-
\end{align}
\end{subequations}
where $\nabla^\pm=\nabla^{\pm\hat A}$, as before.

\begin{remark}
	The Hessian potential is determined up to terms linear in the Hessian coordinates $\eta_k$.
	Indeed, if $\psi,\phi$ are two potentials such that, locally,
	\[
		g_{ij} = \frac{\partial^2\psi}{\partial{\eta_i}\partial{\eta_j}}=\frac{\partial^2\phi}{\partial{\eta_i}\partial{\eta_j}}\,,
	\]
	then
	\[
		\frac{\partial\psi}{\partial\eta_k}=\frac{\partial\phi}{\partial\eta_k}+c_k\,,
	\]
	and hence
	\[
		\psi=\phi+\sum_{k=1}^n c_k\eta_k+c_0\,,
	\]
	for some $c_k\in\RR$.
	This allows one to obtain, from the general solution of~\eqref{eq:hessian}, the Hessian coordinates for each superintegrable system.
\end{remark}

Let $V$ be an affine space, with the same dimension $n$ as $M$.
\emph{Hessian coordinates} for the Hessian potential $\psi$ are defined by a bijective function
\begin{equation*}
	\eta:M\to V
\end{equation*}
for $i\in\{1,\dots,n\}$.
We obtain the \emph{natural Hessian coordinates} for the given Hessian potential $\psi:M\to\RR$ via $\eta^k:=\dot\gamma^k(\psi)$, where $(\gamma^k)$ denotes a geodesic coordinate system on $M$, $\gamma^k\in\mathfrak{X}(M)$.
Now consider two sets of Hessian coordinates $\eta$ and $\theta$, for Hessian potentials $\psi$ and $\phi$, respectively,
\begin{align*}
	\eta &:M\to V \\
	\theta &:M\to W
\end{align*}
where $V$ and $W$ are (different) affine spaces.
Then there is a natural mapping $\langle-,-\rangle:V\times W$ determined via
\[
	\langle v,w\rangle:=g(\eta^*(v),\theta^*(w))
	%g(\eta^{-1}(v),\theta^{-1}(w))
\]
for $v\in V$, $w\in W$, where we identify $V$ and $W$ with their tangent spaces.
We can therefore interpret $W$ as the dual space of $V$.

In the remainder of the section, we discuss the Hessian structures on Riemannian manifolds $(M,g)$ of constant sectional curvature and dimension $n\in\{2,3\}$ that are naturally associated with some well-known non-degenerate systems.
We compute the Hessian potentials by solving the system~\eqref{eq:hessian}. Indeed, we will solve the equations for $\nabla^-=\nabla^{-\hat A}$.
We also compute the Hessian coordinates $(\eta^-_1,\dots,\eta^-_n)$ for each system.
As mentioned earlier, the Hessian potential $\psi_-$ for $\nabla^-$ is determined up to terms linear in the Hessian coordinates $\eta^-_i$ only. Making an (obvious) choice for the integration constants, we hence fix the Hessian potential $\psi_-$ and associated Hessian coordinates~$\eta^-:M\to\RR^n$.

We then obtain the \emph{dual Hessian coordinates} $\eta_+^i$ and the \emph{dual Hessian potential}~$\psi_+$ as follows, imposing requirements that ensure that the dual potential $\psi_+$ is the Legendre-Fenchel transform of $\psi_-$, see~\cite{Amari2009} for example: %AN2000,Matsuzoe2014
\begin{align*}
	\frac{\partial\psi_-}{\partial\eta^-_i} &= \eta_+^i\ \,, &
	\frac{\partial\psi_+}{\partial\eta_+^i} &= \eta^-_i\ \,,
\end{align*}
and
\[
	\psi_++\psi_- = \langle\eta_+,\eta^-\rangle\,.
\]
Concretely, given $\psi_-$ and $\eta^-$, we first compute $\eta_+$ using the above formula. We then set
\[
	\psi_+= -\psi_- + \langle\eta_+,\eta^-\rangle\,.
\]
The remaining condition is redundant and serves as a consistency check.

%\begin{definition}
We say that a Hessian structure is \emph{self-dual} if the dual Hessian potential and the Hessian potential coincide up to linear terms in the Hessian coordinates.
%\end{definition}
%\noindent
Note that this definition is well-posed. Indeed, if we add to $\psi_-$ (irrelevant) linear terms in $\eta^-$, we obtain
\[
	\eta_{+,\text{new}}^i=\frac{\partial}{\partial\eta^-_i}\bigg(
			\underbrace{ \psi_-+\sum_kc_k\eta^-_k }_{=\psi_{-,\text{new}}}
		\bigg)
			=\eta_{+}^i+c_i
\]
and then
\begin{align*}
	\psi_{+,\text{new}}
	&= -\psi_{-,\text{new}} - \langle\eta_{+,\text{new}},\eta^-\rangle \\
	&= \psi_+ -\sum_kc_k\eta^-_k- \sum_k c_k\eta^- 
	= \psi_+ -\sum_k(2c_k)\eta^-_k\,.
\end{align*}
Hence, also $\psi_+$ is modified by linear terms in $\eta^-$ only.

%============================================================================%
\subsection{Hessian coordinates for non-degenerate systems in dimensions two}

We recall that all non-degenerate systems in dimensions two are trivially abundant and that non-degenerate systems in dimension three are classified, c.f.~\cite{Kalnins&Kress&Pogosyan&Miller,KKM05c,KMP00a}.
The $2$-dimensional systems on flat space, i.e.~$(M,g)$ with
\[
	g=dx^2+dy^2,
\]
are the first and second Smorodinski-Winternitz system, (i) respectively (ii), and the non-degnerate harmonic oscillator system (ho), which are defined by the following potentials:
\begin{align*}
	(i) && V&= a_0(x^2+y^2)+\frac{a_1}{x^2}+\frac{a_2}{y^2} + a_3  && (x\ne0,y\ne0)
	\\
	(ii) && V&= a_0(4x^2+y^2)+a_1x+\frac{a_2}{y^2} + a_3 && (y\ne0)
	\\
	(ho) && V &= a_0(x^2+y^2)+a_1x+a_2y + a_3\,.
\end{align*}
We also consider a system on the round $2$-sphere, i.e.\ for the metric
\[
	\qquad\qquad g=dX^2+dY^2+dZ^2\quad \text{  (with $X^2+Y^2+Z^2=1$)},
\]
where $(X,Y,Z)$ are standard coordinates on $\RR^3$, namely the so-called \emph{generic system} defined by the potential $(X\ne0,Y\ne0,Z\ne0)$
\[
	V= \frac{a_0}{X^2}+\frac{a_1}{Y^2}+\frac{a_2}{Z^2}+a_3\,,
\]
where $a_i\in\RR$. This system plays a pivotal role in the classification of $2$-dimensional second-order (maximally) superintegrable systems, c.f.~\cite{KMP13,KKM18}.

\subsubsection{The $2$-dimensional harmonic oscillator system (ho)}
For the harmonic oscillator system, solving~\eqref{eq:wilczynski} yields $A=0$ by virtue of~\eqref{eq:A.from.T}.
By direct integration of~\eqref{eq:A.hessian}, we then obtain the Hessian potential as $\psi_-\in\Psi$, where
\[ \Psi = \frac12(x^2+y^2)+\Span(x,y,1)\,. \]
Here and in the following, $\Span$ denotes an arbitrary linear combination of the functions specified as arguments.
As explained earlier, we may choose the Hessian potential to be
\[
	\psi_-=\frac12(x^2+y^2)
\]
and the Hessian coordinates as
\[
	\eta^-_1=x\,,\qquad\eta^-_2=y\,.
\]
%Note that this choice is made for convenience and without loss of generality.
The dual Hessian coordinates are thus obtained as
\[
	\eta_+^1=\frac{\partial\psi_-}{\partial\eta^-_1}=x\quad\text{and}\quad
	\eta_+^2=\frac{\partial\psi_-}{\partial\eta^-_2}=y\,.
\]
We hence obtain
\[
	\psi_+=-\psi_-+\langle\eta^-,\eta_+\rangle = -\frac12(x^2+y^2)+x^2+y^2=\frac12(x^2+y^2)=\psi_-\,.
\]
We therefore conclude that the Hessian structure is \emph{self-dual}.

\subsubsection{The Smorodinski-Winternitz system (i)}
Solving~\eqref{eq:wilczynski} for the Smorodinski-Winternitz family (i) of potentials and using~\eqref{eq:A.from.T}, yields $(x\ne0,y\ne0)$
\[
	A = -\frac1x\,dx^3-\frac1y\,dy^3\,.
\]
By direct integration of~\eqref{eq:A.hessian}, we then obtain the Hessian potential as $\psi_-\in\Psi$, where
\[ \Psi = \frac12(y^2\ln|y|+x^2\ln|x|)+\Span(y^2,x^2,1) \]
Choosing
\[
	\psi_-=\frac12(y^2\ln|y|+x^2\ln|x|)\,,\qquad
	\eta^-=(x^2,y^2)\,,
\]
we then obtain
\[
	\psi_-=\frac12\,\left( \eta^-_1\ln\sqrt{\eta^-_1}+\eta^-_2\ln\sqrt{\eta^-_2} \right)
	=\frac14\,\left( \eta^-_1\ln\eta^-_1+\eta^-_2\ln\eta^-_2 \right)
\]
and hence
\[
	\eta_+=\left(\frac{\ln\eta^-_1}{4}+\frac14\,,\,\frac{\ln\eta^-_2}{4}+\frac14\right)
	=\left(\frac{\ln|x|}{2}+\frac14\,,\,\frac{\ln|y|}{2}+\frac14\right)\,.
\]
Next, we obtain
\[
	\langle\eta^-,\eta_+\rangle = \psi_-+\frac14(x^2+y^2)
\]
This implies $\psi_+=-\psi_-+\langle\eta^-,\eta_+\rangle = \frac14(x^2+y^2)$ and hence
\[
	\psi_+=\frac14\left( e^{2\eta_+^1-\frac12}+e^{2\eta_+^2-\frac12}\right)\,.
\]
We therefore conclude that the Hessian structure is \emph{not} self-dual.

%The Hessian coordinates are obtained as $(\xi,\eta)=(x^2,y^2)$. In these coordinates, the Hamiltonian reads (for $x=\sqrt\xi$ and $y=\sqrt\eta$)
%\[
%	H=4\xi\,p_\xi^2+4\eta\,p_\eta^2+a_0(\xi+\eta)+\frac{a_1}{\xi}+\frac{a_2}{\eta}+a_3
%\]

\subsubsection{The Smorodinski-Winternitz system (ii)}
Solving~\eqref{eq:wilczynski} for the Smorodinski-Winternitz family (ii) of potentials and using~\eqref{eq:A.from.T}, yields $(y\ne0)$
\[
	A = -\frac1y\,dy^3\,.
\]
By direct integration of~\eqref{eq:A.hessian}, we then obtain the Hessian potential as $\psi_-\in\Psi$, where
\[ \Psi = \frac12(x^2+y^2\ln|y|) + \Span(x,y^2,1). \]
Choosing
\[
	\psi_-=\frac12(x^2+y^2\ln|y|)\,,\qquad
	\eta^-=(x,y^2)\,,
\]
we obtain
\[
	\psi_-=\frac12\,\left( (\eta^-_1)^2+\eta^-_2\ln\sqrt{\eta^-_2} \right)
	=\frac12\,(\eta^-_1)^2+\frac14\eta^-_2\ln\eta^-_2
\]
and hence
\[
	\eta_+=\left(\eta^-_1\,,\,\frac{\ln\eta^-_2}{4}+\frac14\right)
	=\left(x\,,\,\frac{\ln|y|}{2}+\frac14\right)\,.
\]
Next, we obtain
\[
	\langle\eta^-,\eta_+\rangle = x^2+\frac12y^2\ln|y|+\frac14y^2\,.
\]
This implies $\psi_+=-\psi_-+\langle\eta^-,\eta_+\rangle = \frac12x^2+\frac14y^2$ and hence
\[
	\psi_+=\frac12x^2+\frac14 e^{2\eta_+^2-\frac12}\,.
\]
We therefore conclude that the Hessian structure is \emph{not} self-dual.

\subsubsection{The 2-sphere (s2)}
Solving~\eqref{eq:wilczynski} and~\eqref{eq:A.from.T} for the system on the round sphere, and by a direct integration of~\eqref{eq:A.hessian}, we obtain the Hessian potential as $\psi_-\in\Psi$, where $(X\ne0,Y\ne0,Z\ne0)$
\[
	\Psi=\frac12\left( X^2\ln|X|+Y^2\ln|Y|+Z^2\ln|Z| \right)+\Span(X^2,Y^2,Z^2)
\]
subject to $X^2+Y^2+Z^2=1$.
%The geodesics of $S^2$ lie along intersections of $X^2+Y^2+Z^2=1$ with the planes
%\[
%	b_1X+b_2Y+b_3Z = 0\,.
%\]
We may choose the Hessian coordinates as
\[
	\eta^-=\left(X^2,Y^2\right)\,,
\]
and hence
\[
	\psi_-=\frac14\left(
		(\eta^-_1)\ln(\eta^-_1)+(\eta^-_2)\ln(\eta^-_2)+(1-\eta^-_1-\eta^-_2)\ln(1-\eta^-_1-\eta^-_2)
	\right)
\]
We obtain for the dual Hessian coordinates the expressions
\[
	\eta_+^1=\frac14\ln\frac{\eta^-_1}{1-\eta^-_1-\eta^-_2}\,,\quad
	\eta_+^2=\frac14\ln\frac{\eta^-_2}{1-\eta^-_1-\eta^-_2}\,,
\]
We therefore obtain
\[
	\psi_-=\frac14\left(
		\eta^-_1\eta_+^1+\eta^-_2\eta_+^2+\ln(1-\eta^-_1-\eta^-_2)
	\right)
\]
\[
	\psi_+
%	=-\psi_--\frac14\left(
%		\eta^-_1\ln\frac{\eta^-_1}{1-\eta^-_1-\eta^-_2}
%		+\eta^-_2\ln\frac{\eta^-_2}{1-\eta^-_1-\eta^-_2}
%	\right)
	=-\psi_-	+\eta^-_1\eta_+^1	+\eta^-_2\eta_+^2
	=-\frac14\ln(1-\eta^-_1-\eta^-_2)
	=-\frac12\ln(Z)
\]
We conclude that the Hessian structure associated with the 3-sphere system is \emph{not} self-dual.

%\subsubsection{Level surfaces}
%
%The level surfaces of Hessian potentials are affine hypersurfaces, see \cite{shima}. We denote by $\mathcal W$ the Lambert-W function.\medskip
%
%\begin{tabular}{l|l|l}
%	System
%	& Hessian potential $\psi_-$
%	& level surface $\{\psi_-=k\}$ 
%	%dual Hessian potential $\psi_+$
%	%& level surface $\{\psi_+=k\}$
%	\\
%	\hline
%	(ho)
%	& $x^2+y^2$
%	& Circle (1-sphere) ($k>0$)
%	\\
%	\hline
%	(i)
%	& $\frac14(\eta^-_1\ln\eta^-_1+\eta^-_2\ln\eta^-_2)$
%	& $y=\frac{1-x\ln(x)-k}{W(-x\ln(x)+k)}$ ($y\ne1$)
%	\\
%	\hline
%	(ii)
%	& 
%	& $x=\sqrt{k-y\ln(y)}$ ($x>0$)
%	%
%	%sphere: dual potential has Z=const as level surfaces (small circles)
%\end{tabular}

%============================================================================%
\subsection{Hessian coordinates for non-degenerate systems in dimensions three}

We recall that all non-degenerate systems in dimension three are abundant, and that they are classified, compare~\cite{KKM05c,KKM07c,Capel_phdthesis,Capel}.
We shall discuss the systems defined by the potentials
\begin{align*}
	(I) && V&= a_0(x^2+y^2+z^2)+\frac{a_1}{x^2}+\frac{a_2}{y^2}+\frac{a_3}{z^2} + a_4
	& (x\ne0,y\ne0,z\ne0)
	\\
	(IIa) && V&= a_0(4x^2+4y^2+z^2)+a_1x+a_2y+\frac{a_3}{z^2} + a_4
	& (z\ne0)
	\\
	(IIb) && V&= a_0(4x^2+y^2+z^2)+a_1x+\frac{a_2}{y^2}+\frac{a_3}{z^2} + a_4
	& (y\ne0,z\ne0)
	\\
	(HO) && V &= a_0(x^2+y^2+z^2)+a_1x+a_2y+a_3z + a_4
	\\
	\intertext{where the underlying metric is $g=dx^2+dy^2+dz^2$. We also consider the so-called \emph{generic system} on the $3$-sphere, defined by the potential}
	(S3) && V&= \frac{a_0}{X^2}+\frac{a_1}{Y^2}+\frac{a_2}{Z^2}+\frac{a_3}{W^2}+a_4
	&
\end{align*}
where $(X,Y,Z,W)$ are standard coordinates on $\RR^4$ $(X\ne0,Y\ne0,Z\ne0,W\ne0)$.
The underlying metric in the case (S3) is $g=dX^2+dY^2+dZ^2+dW^2$ restricted to $X^2+Y^2+Z^2+W^2=1$, i.e.\ the round metric. Like its $2$-dimensional counterpart (s2), the generic system plays a pivotal role in the theory of $3$-dimensional second-order (maximally) superintegrable systems \cite{CKP15,KKM18}.
The mentioned examples (I)--(S3) fall under the hypotheses of Theorem~\ref{thm:main}.

\subsubsection{The Harmonic Oscillator (HO)}
Analogously to the 2-dimensional case, we obtain $A=0$. We then find for the Hessian potential $\psi_-$ that $\psi_-\in\Psi$, where
\[ \Psi = \frac12(x^2+y^2+z^2) + \Span(x,y,z,1) \]
Choosing
\[
	\psi_- = \frac12(x^2+y^2+z^2)\,,
\]
and the Hessian coordinates
\[
	\eta^-=(x,y,z)\,,
\]
the dual Hessian coordinates are straightforwardly found to be
\[
	\eta_+=(x,y,z)\,.
\]
The dual Hessian potential hence is
\[
	\psi_+=\psi_-=\frac12(x^2+y^2+z^2)\,,
\]
yielding that the Hessian structure associated to the 3-dimensional harmonic oscillator system is \emph{self-dual}.

\subsubsection{The Smorodinski-Winternitz system (I)}
Analogously to the 2-dimensional case, we obtain $(x\ne0,y\ne0,z\ne0)$
\[
	A = -\frac1x\,dx^3-\frac1y\,dy^3-\frac1z\,dz^3\,.
\]
and then find that $\psi_-\in\Psi$, where
\[ \Psi = \frac12\left(x^2\ln|x|+y^2\ln|y|+z^2\ln|z|\right) + \Span(x^2,y^2,z^2,1) \]
Choosing
\[
	\psi_- = \frac12\left(x^2\ln|x|+y^2\ln|y|+z^2\ln|z|\right)\,,
\]
and the Hessian coordinates
\[
	\eta^-=(x^2,y^2,z^2)\,,
\]
the dual Hessian coordinates are straightforwardly found to be
\[
	\eta_+=\left( \,\frac{\ln|x|}{2}+\frac14\,,\ \frac{\ln|y|}{2}+\frac14 \,,\ \frac{\ln|z|}{2}+\frac14\, \right)\,.
\]
The dual Hessian potential hence is
\[
	\psi_+=\frac14\left(e^{2\eta_+^1-\frac12} + e^{2\eta_+^2-\frac12} + e^{2\eta_+^3-\frac12}\right)\,,
\]
yielding that the Hessian structure associated to the 3-dimensional harmonic oscillator system is \emph{not} self-dual.

\subsubsection{The Smorodinski-Winternitz system (IIa)}
We obtain $(z\ne0)$
\[
	A = -\frac1z\,dz^3\,,
\]
and then find $\psi_-\in\Psi$, where
\[ \Psi = \frac12(x^2+y^2+z^2\ln|z|) + \Span(x,y,z^2,1)\,, \]
and choose the Hessian potential
\[
	\psi_- = \frac12(x^2+y^2+z^2\ln|z|)
\]
and the Hessian coordinates
\[
	\eta^- = (x,y,z^2)\,.
\]
We thus find, in a manner entirely analogous to the previous cases, the dual Hessian coordinates
\[
	\eta_+ = \left(x\,,\ y\,,\ \frac{\ln|z|}{2}+\frac14\right)
\]
and the dual Hessian potential
\[
	\psi_+ = \frac12(x^2+y^2)+\frac14\,e^{2\eta_+^3-\frac12}\,.
\]
The Hessian structure is not self-dual.

\subsubsection{The Smorodinski-Winternitz system (IIb)}
We obtain $(y\ne0,z\ne0)$
\[
	A = -\frac1y\,dy^3-\frac1z\,dz^3\,,
\]
and then find $\psi_-\in\Psi$, where
\[ \Psi = \frac12(x^2+y^2\ln|y|+z^2\ln|z|) + \Span(x,y^2,z^2,1)\,. \]
We choose the Hessian potential
\[
	\psi_- = \frac12(x^2+y^2\ln|y|+z^2\ln|z|)
\]
and the Hessian coordinates
\[
	\eta^- = (x,y^2,z^2)\,.
\]
In these coordinates, the Hessian potential takes the form
\[
	\psi_- = \frac12((\eta^-_1)^2+\frac12\eta^-_2\ln(\eta^-_2)+\frac12\eta^-_3\ln(\eta^-_3))
\]
We then obtain, in an entirely analogous manner to the previous cases, the dual hessian coordinates
\[
	\eta_+ = \left(x\,,\ \frac{\ln|y|}{2}+\frac14\,,\ \frac{\ln|z|}{2}+\frac14\right)
\]
and the dual Hessian potential
\[
	\psi_+ = \frac12x^2+\frac14\left( e^{2\eta_+^2-\frac12} + e^{2\eta_+^3-\frac12} \right)\,.
\]
The Hessian structure is not self-dual.

\subsubsection{The 3-sphere (S3)}
We find, analogously to the 2-dimensional case, $\psi_-\in\Psi$,
\[ \Phi = \frac12(X^2\ln|X|+Y^2\ln|Y|+Z^2\ln|Z|+W^2\ln|W|) +\Span(X^2,Y^2,Z^2,W^2,1)\,, \]
subject to $X^2+Y^2+Z^2+W^2=1$, $X\ne0,Y\ne0,Z\ne0,W\ne0$, and we then choose the Hessian potential
\[
	\psi_- = \frac12(X^2\ln|X|+Y^2\ln|Y|+Z^2\ln|Z|+W^2\ln|W|)
\]
as well as the Hessian coordinates
\[
	\eta^-=(X^2,Y^2,Z^2)\,.
\]
The dual Hessian coordinates are then obtained as
\[
	\eta_+ = \frac14\left(
				\ln\frac{\eta^-_1}{1-\eta^-_1-\eta^-_2-\eta^-_3}\,,\ 
				\ln\frac{\eta^-_2}{1-\eta^-_1-\eta^-_2-\eta^-_3}\,,\ 
				\ln\frac{\eta^-_3}{1-\eta^-_1-\eta^-_2-\eta^-_3}
			 \right)\,.
\]
The dual Hessian potential hence takes the form
\[
	\psi_+=-\frac12\ln|W|\,.
\]
The Hessian structure associated to the 3-sphere therefore is not self-dual.

%============================================================================%
\subsubsection{Self-duality for the Hessian potentials of flat abundant systems}

The discussion in the previous paragraphs implies that for the harmonic oscillator in arbitrary dimension we have
$A=0$,
and hence
\[
	\frac{\partial^3\psi}{\partial\eta_i\partial\eta_j\partial\eta_k} = 0
\]
for all $i,j,k$. Here $\psi$ can be either $\psi_-$ or $\psi_+$, and $\eta=(\eta_1,\dots,\eta_n)$ are understood to be the corresponding Hessian coordinates.
It follows \cite[\S~5]{Shima} that $\psi$ is a polynomial of degree two in $\eta$, or, equivalently, that its first Koszul form vanishes \cite[Cor.~5.4]{Shima}. Moreover, the vanishing of the first Koszul form implies that the Levi-Civita connection $\nabla$ of $g$ coincides with the flat connection $D^+$, respectively~$D^-$, which in turn implies $A=0$.
We therefore conclude:
\begin{proposition}
	Consider an abundant superintegrable system on a manifold $(M,g)$ of constant sectional curvature, and denote by $D$ (either of) the naturally associated flat connections. The first Koszul form of $(M,D,g=Dd\psi)$ vanishes, if and only if the system is the non-degenerate $n$-dimensional harmonic oscillator on~$M$.
	In this case, $(M,g)$ is flat.
\end{proposition} 

\noindent To be more specific, the associated cubic of the harmonic oscillator is $A=0$. Hence $\nabla^3\psi=0$, $\nabla=D^+=D^-=:D$.
This is the case if and only if, with flat coordinates $x_k$, $\psi=a\sum_k x_k^2+\sum b_kx_k+c$ for constants $a,b_j,c\in\RR$.
We conclude, since $\nabla^2\psi=D^2\psi=g=\sum_k dx_k^2$, that $a=\frac12$ and then decree (using the available freedom)
\[
	\eta^-=(x_1,\dots,x_n)
\]
for the Hessian coordinates, as well as
\[
	\psi_-=\frac12\sum x_k^2
\]
for the Hessian potential.
A computation perfectly analogous to that for the cases (ho) and (HO) then confirms the self-duality of this Hessian structure, and, in particular, $\eta_+=(x_1,\dots,x_n)$, $\psi_+=\psi_-$.

%============================================================================%
\subsubsection{Hessian potentials and structure functions for abundant systems}

We recall that in~\cite{KSV2023} it is shown that, for abundant systems of dimension $n\geq3$ on spaces of constant sectional curvature, there exists a function $\phi$, called \emph{structure function}, which satisfies
\begin{equation}\label{eq:A.via.structure.function}
	A = -\frac13\left( \nabla^3\phi + 4\kappa\, g\otimes d\phi \right)\,
\end{equation}
and is a priori independent from the Hessian potential.
A list of the structure functions in dimension $n=3$ can be found in \cite{KSV2023}. Note that, like Hessian potentials, also structure functions are not unique and subject to a gauge freedom, c.f.~\cite{KSV2023}.

\begin{observation}
	Comparing the structure functions of non-degenerate systems on Riemannian manifolds of constant sectional curvature in dimension $n=3$ with the Hessian potentials and dual Hessian potentials found earlier in this chapter, we observe that the structure functions coincide, modulo gauge freedom and modulo a factor of $-3$, with the Hessian potentials of their canonical Hessian structures, 
	$$ \phi=-3\psi_-\quad\text{modulo gauge freedom}\,. $$
\end{observation}
The aim of the present section is to elucidate this observation.
To this end, we write $\psi=\psi_-$ and $D=\nabla^{-A}$ for brevity, and denote the Levi-Civita connection of~$g$ by~$\nabla$.
We then compute, for arbitrary $X,Y,Z\in\mathfrak X(M)$,
\begin{align*}
	D^3\psi(X,Y,Z)
	&= D_Z(Dd\psi)(X,Y)
	\\
	&= Z((Dd\psi)(X,Y))-Dd\psi(D_ZX,Y)-Dd\psi(X,D_ZY)
	\\
	&= Z(\nabla d\psi(X,Y)-\hat A(X,Y)(d\psi))
	\\ &\qquad\qquad -g(\nabla_ZX,Y)-g(X,\nabla_ZY)-2A(X,Y,Z)
	\\
	&= \nabla^3\psi(X,Y,Z)-\nabla_Z\hat A(X,Y)(d\psi)
	\\ &\qquad\qquad -\hat A(X,Y)(D_Zd\psi-\hat A(Z,\cdot)(d\psi))-2A(X,Y,Z)
	\\
	&= \nabla^3\psi(X,Y,Z) - A(X,Y,Z)
	\\ &\qquad\qquad +(\nabla_Z\hat A)(X,Y)(d\psi)+\hat A(\hat A(X,Y),Z)(d\psi)
\end{align*}
%In the computation we have, without loss of generality, assumed that $X,Y$ and $Z$ are parallel extensions with respect to $\nabla$, i.e.\ we have tacitly performed a pointwise computation.
%
For abundant systems in dimension $n\geq3$, there is a formula expressing the derivatives of $A$ in terms of an quadratic algebraic expression in $A$, see (5.10) of~\cite{KSV2023}. In terms of $A$, we may write this expression as
\begin{align*}
	\nabla_lA_{ijk} = \frac49\Pi_{(ijkl)}S_{ijl}\bar t_k + \Pi_{(ijk)}
	&\smash{\bigg(}
		3g_{ij}\mathsf{P}_{kl}
		+\frac19 S\indices{_{ij}^a}S_{kla}
	\\ &\quad
		+\frac4{9(n-2)} \mathring{\mathscr{S}}_{ik} g_{jl}
		+\frac1{9(n-2)} g_{ij}\mathring{\mathscr{S}}_{kl}
	\\ &\quad
		-\frac13 S_{ija}\bar t^a g_{kl}
		+\frac13 \,g_{jk}\,\bar t_{i}\bar t_{l}
		-\frac16 |\bar t|^2\ g_{ij}g_{kl}
	\\ &\quad
		-\frac{5}{18n(n-1)}\ g_{ij}g_{kl}\,\mathrm{trace}_g(\mathscr{S})
	\smash{\bigg)}\,,
\end{align*}
where $\mathsf{P}$ denotes the Schouten tensor of $g$ and where $\mathring{\mathscr S}$ is the trace-free part of the tensor field $\mathscr S$ defined by
\[
	\mathscr S(X,Y)=\mathrm{trace}_g\left(\hat S(X,\hat S(Y,-))\right)\,,
\]
where $\hat S=g^{-1}S$.
Moreover, c.f.~(5.11) of~\cite{KSV2023},
\[
	\Pi_{\textrm{Weyl}} \left( A(-,-,\hat A(-,-)) \right) = 0\,,
\]
where $\Pi_{\mathrm{Weyl}}$ is the projector onto the Weyl symmetric component, and
\[
	\hat S(\bar t)+\bar t\otimes\bar t-\frac1n\,|\bar t|^2\,g-\frac{1}{n-2}\,\mathring{\mathscr S} = 3\,\mathring{\mathsf P}
\]
where $\mathring{\mathsf P}$ denotes the trace-free part of the Schouten tensor of $g$.
Finally, as the underlying metric $g$ is of constant sectional curvature $\kappa$, we have $\mathring{\mathsf P}=0$ and hence
\[
	\mathrm{trace}_g(\mathscr{S})-(n+2)(n-1)\,|\bar t|^2=9n(n-1)\kappa\,,
\]
c.f.~(7.10c) of~\cite{KSV2023}. Equipped with these facts, we can prove the following claim.
\begin{proposition}
	Let $(M^n,g)$ be of constant sectional curvature $\kappa$ and dimension $n\geq3$.
	Then the Hessian potential $\psi_-$ of an abundant second-order superintegrable system on $(M,g)$ satisfies
	\[
		A = \nabla^3\psi_- + 4\kappa\, g\otimes d\psi_-
	\]
\end{proposition}
\begin{proof}
	We substitute $\nabla A$ and $\hat A(\hat A(X,Y),Z)$ in the expression for $D^3\psi_-$ obtained earlier in this chapter.
	Hence,
	\begin{align*}
		-A(X,Y,Z) &= \nabla^3\psi_-(X,Y,Z) +(\nabla_Z\hat A)(X,Y)(d\psi_-)+\hat A(\hat A(X,Y),Z)(d\psi_-)
		\\
		&= \nabla^3\psi_-(X,Y,Z) + 4\kappa\, g(X,Y) Z(\psi_-)
	\end{align*}
	Comparing this result with \eqref{eq:A.via.structure.function},
	we find
	\[
			-3\left(
				\nabla^3\psi_- + 4\kappa\, g\otimes d\psi_-
			\right) = \nabla^3\phi + 4\kappa\, g\otimes d\phi\,.
	\]
	and hence conclude $\phi=-3\psi_-$, up to the obvious gauge freedom.
\end{proof}

%============================================================================%
\subsection{Discussion and Outlook}

In this paper we have shown that abundant second-order superintegrable systems on Riemannian manifolds of constant sectional curvature carry a natural Hessian structure. For the abundant systems defined on flat space, the flat connection of the natural Hessian structures  appear to not differ much from the Levi-Civita connection, but they do encode the abundant system entirely. 
The class of non-degenerate systems is by definition broader than that of abundant systems, but to date no non-abundant non-degenerate systems have been found. There has therefore been some speculation that such example might not exist. Yet, the conditions~\eqref{eqs:hessian.conditions} appear naturally, suggesting that either the conditions for non-degenerate systems encode stronger restrictions which are not visible straightforwardly, or that non-abundant non-degenerate systems have to exist.

\section*{Acknowledgements}
The authors acknowledge support from the DFG project grant \#540196982 and the \emph{Forschungsfonds} of the University of Hamburg.
Andreas Vollmer is grateful to Benjamin McMillan and also thanks Vicente Cortés for discussions.

%============================================================================%
%\clearpage
\bibliographystyle{amsalpha}
\bibliography{superintegrable-hessian}
\medskip

\end{document}